\documentclass{amsart}
\usepackage{amscd,amssymb,subfigure,hyperref, epsfig}
\usepackage[arrow,matrix,graph,frame,poly,arc,tips]{xy}
\usepackage{amsmath}

\theoremstyle{plain}
\newtheorem{theorem}{Theorem}

\theoremstyle{definition}

\newtheorem{definition}{Definition}
\newtheorem{example}{Example}

\def\F{\mathbb{F}}
\def\bFq{\F_q}
\def\proj{\mathbb{P}}
\def\G{\mathbb{G}}
\def\C{\mathbb{C}}
\def\Z{\mathbb{Z}}
\newcommand\smallda{ \hskip0.04truein{\hbox{--}}\hskip0.06truein{\hbox{--}} \to}
\begin{document}

\title[Codes from Higher Dimensional Varieties]{Algebraic Geometry Codes
from Higher Dimensional Varieties\label{ChLittle}}

\author[John B. Little]{John B. Little}

\address{Department of Mathematics and Computer Science\\
College of the Holy Cross\\
Worcester, MA 01610 USA\\
{\tt little@mathcs.holycross.edu}}



\begin{abstract}
This paper is a general survey of work on Goppa-type codes
from higher dimensional algebraic varieties.  The
construction and several techniques for estimating the minimum
distance are described first.  Codes from various classes of
varieties, including Hermitian hypersurfaces, Grassmannians,
flag varieties, ruled surfaces over curves, and Deligne-Lusztig
varieties are considered.  Connections with the theories of toric
codes and order domains are also briefly indicated.
\end{abstract}
\maketitle

\section{Introduction}\label{introsec}

The codes considered in this survey can all be understood
as examples of \emph{evaluation codes} produced from a finite
set ${\mathcal S} = \{P_1,\ldots,P_n\}$ of $\bFq$-rational
points on an algebraic variety $X$ and an $\bFq$-vector space of 
functions ${\mathcal F}$ defined on ${\mathcal S}$.  The set of codewords 
is the image of an evaluation mapping
\begin{eqnarray}\label{ev}
ev_{\mathcal S} : {\mathcal F} & \longrightarrow & \bFq^n\\
                f &\mapsto & (f(P_1),\ldots,f(P_n)).\nonumber
\end{eqnarray}
$X$ will usually be assumed smooth, but
in fact many of the constructions also make sense for 
normal varieties (much of the usual geometric
theory of divisors and line bundles on normal varieties
is the same as in the smooth case).  

The Goppa $C_L(D,G)$ codes from curves $X$ where ${\mathcal F} = L(G)$
for some divisor $G$ on $X$ were the first 
examples of codes of this type to be considered.  
Relatively early in the history of applications of 
algebraic geometry to coding theory, however, Tsfasman and Vladut 
proposed in Chapter 3.1 of \cite{tv91} that higher
dimensional varieties might also be used to construct 
codes.  By the results of \cite{psw91}, \emph{every} linear code 
can be obtained by the construction of Definition~\ref{sophis} below,
starting from some ${\mathcal S} \subseteq X(\bFq)$ for some variety $X$
and some line bundle ${\mathcal L}$ on $X$;  indeed \emph{curves} 
suffice for this (see Section \ref{bibnotes}).
Hence the question is whether one can identify specific 
higher dimensional varieties
$X$, spaces of functions ${\mathcal F}$, and sets 
of rational points ${\mathcal S}$ 
that yield particularly interesting codes
using algebraic geometric constructions.
There has been a fairly steady stream of articles
since the 1990's studying such codes and our first main goal 
here is to survey the methods that have been developed
and the results that have been obtained.  

In a sense, the first major difference between 
higher dimensional varieties and curves
is that points on $X$ of dimension $\ge 2$ are subvarieties
of codimension $\ge 2$, not divisors.  This means that
many of the familiar tools used for Goppa codes (e.g. 
Riemann-Roch theorems, the theory of differentials and
residues, etc.) do not apply in exactly the same way.

A second difference is the possibility of performing
\emph{birational modifications} such as blowing up points or
other subvarieties on a variety of higher dimension.  
For instance, if $p$ is a point in a smooth algebraic
variety $X$ of dimension $\delta \ge 2$,
there is another smooth variety $Y = {\rm Bl}_p(X)$, a
proper morphism $\pi : Y \to X$, and an \emph{exceptional divisor}
$E \simeq \proj^{\delta - 1}$ in $Y$ such that
$\pi(E) = \{p\}$, and $\pi|_{Y-E} : Y - E \simeq X - \{p\}$ as
varieties.  Because $Y$ and $X$ have isomorphic nonempty Zariski-open
subsets, they have isomorphic function fields.  
Such varieties $Y$ and $X$ are said to be 
\emph{birationally isomorphic}.
This says that function fields in two or more variables
always have many different nonisomorphic smooth models, and
the connection with function fields is not as tight
as in the curve case.

It must be said that the theory of Goppa-type
codes from higher dimensional varieties is much less advanced
at this point than the theory for Goppa codes from curves, perhaps
because of these differences.  There is still no clear understanding
of how best to harness the properties of higher dimensional
varieties in coding theory.
Indeed, as we will see, most of the work that has appeared
to date has been devoted to case studies of
the \emph{structural properties} of codes constructed from 
certain particular families of varieties $X$ -- their parameters, 
their weight distributions, their hierarchies of higher 
Hamming weights, and so forth.  A few general ideas for 
estimating the minimum distance $d$ have been developed.  However, 
quite a few of the codes that we will see are rather unremarkable; 
in many of the cases where the exact weight distributions are 
known, other algebraic constructions yield better codes. 
In addition, the development of efficient encoding and 
decoding algorithms for these codes has not really begun (see 
Section~\ref{bibnotes} on this point, though). 
The theory of order domains should yield tools 
here as well as for codes from curves.  Nevertheless, the universality
of this construction offers hope that good examples can
be constructed this way, and our second main goal is 
to encourage others to explore this area.  

This survey is organized as follows.  In Section~\ref{constr},
we give two variants of Tsfasman and Vladut's code construction,
one starting from an abstract variety $X$ and line bundle ${\mathcal L}$
on $X$, the other starting from an embedded variety $X \subset \proj^m$.
We also present some first examples.  Four general methods for 
estimating the minimum distance are presented in Section~\ref{params}.
Two appeared first in S.H.~Hansen's article \cite{sh01}.  For the 
first of these, it is assumed that all of the $\bFq$-rational
points of interest are contained in a family of curves on $X$
and intersection products of divisors with those curves are
used to bound $d$.  The second method is based on 
the Seshadri constant of the line bundle ${\mathcal L}$ with
respect to the set of $\bFq$-rational points on $X$.  A 
third method from \cite{gls05} 
can be used when the set of $\bFq$-rational points
is itself a complete intersection in $\proj^m$.  Finally, 
we present another, more arithmetic, method based on the 
Weil conjectures developed by Lachaud in \cite{l96}.

The next sections \ref{exs} and \ref{DLvars} present 
a selection of the examples of these codes that have appeared
in the literature, codes constructed from quadric hypersurfaces,
Hermitian hypersurfaces, Grassmannians and flag varieties, 
Del Pezzo surfaces, ruled surfaces, and Deligne-Lusztig
varieties.  Finally, we present some comparisons between
codes in Section~\ref{comps}.

Where practicable, we have provided brief proofs of the 
results we state, in order to show the methods involved
in the study of these codes.  

As we proceed through these examples, the 
prerequisites from algebraic geometry steadily increase.
Our intended audience includes both coding theorists familiar
with the theory of Goppa codes on curves but not higher dimensional
geometry and algebraic geometers curious about how higher dimensional 
varieties might be used in the coding theory context.  Hence there 
are probably portions of what we say that might seem
unnecessarily elementary to some readers.  
We apologize in advance.

The text \cite{har77} by Hartshorne is a good general reference
for most of the algebraic geometry we need.  The construction
of Grassmannians via exterior algebra, Schubert varieties, 
and the intersection theory on Grassmannians are covered in
Griffiths and Harris, \cite{gh78}.
A full understanding of the Deligne-Lusztig
varieties also depends on the theory of reductive algebraic groups $G$
over fields of characteristic $p$ and the classification
of their finite subgroups $G^F$
by root systems and Dynkin diagrams
with an action of the Frobenius endomorphism, $F$.
The book \cite{ca85} of Carter contains all the information
needed for this.     

Because of space limitations, it has not been possible
to discuss all the results of every paper in this area
in detail.  Pointers to all of the literature 
of which the author is aware are provided in the
bibliographic notes in Section~\ref{bibnotes}, 
the references, and their bibliographies.  

Any omissions or errors are entirely due to the author.  
Any comments or suggestions are welcome.

\subsection{Notation}

We will use the following general notational and terminological
conventions.

\begin{itemize}
\item The number of elements in a finite set ${\mathcal T}$ will be
denoted by $\# {\mathcal T}$.

\item The \emph{parameters} of a linear code are denoted $[n,k,d]$ as
usual, where $n$ is the block length, $k$ is the dimension,
and $d$ is the minimum distance.  

\item The \emph{generalized Hamming weights} are denoted $d_r$, $1 \le r \le k$.
As in \cite{wei91}, $d_r$ is the size of the minimal support of
an $r$-dimensional subcode of $C$, extending the usual minimum distance $d = d_1$.

\item We denote an algebraically closed field of characteristic
$p$ by $\F$ and all finite fields $\bFq$ for $q = p^m$ 
are considered as subfields of $\F$.  

\item The \emph{projective spaces} $\proj^m$, \emph{Grassmannians} $\G(\ell,m)$,
and so forth are considered as varieties over the algebraically
closed field $\F$ in order to ``do geometry.'' The $\bFq$-rational
points used in the construction of the codes are finite subsets
of these varieties.  

\item If $f$ is a homogeneous polynomial in $\bFq[x_0,\ldots,x_m]$, 
${\bf V}(f)$ is the zero locus of $f$ in $\proj^m$.  

\item A \emph{line bundle} is a locally free sheaf of rank one.  
At several points, it will be convenient to use the 
\emph{sheaf cohomology} groups $H^i(X,{\mathcal L})$ for 
a line bundle ${\mathcal L}$.  The space of global sections will 
also be written $\Gamma(X,{\mathcal L})$.  
\end{itemize}

\section{The General Construction}\label{constr}

Several apparently different, but essentially equivalent,
versions of the construction are commonly encountered in
the literature.  For instance, one description starts from a smooth
projective variety $X$ defined over $\bFq$, a set
${\mathcal S} \subseteq X(\bFq)$ of $\bFq$-rational points of $X$,
and a line bundle ${\mathcal L}$ on $X$, also defined over $\bFq$.  
Let $P$ be an $\bFq$-rational point of $X$.  The stalk 
${\mathcal L}_P$, modulo sections vanishing at $P$, denoted
$\overline{{\mathcal L}}_P$, is isomorphic to $\bFq$ by a choice
of local trivialization.  

\begin{definition}\label{sophis}
The choice of such local trivializations at each point in ${\mathcal S}$ 
defines a linear mapping (called the \emph{germ map} in \cite{tv91}) 
\begin{equation}\label{sheafmap}
\alpha : \Gamma(X,{\mathcal L}) \longrightarrow \bigoplus_{i=1}^n \overline{{\mathcal L}}_{P_i} \simeq \bFq^n,
\end{equation}
and the image is the
code denoted $C(X,{\mathcal L};{\mathcal S})$, or $C(X,{\mathcal L})$ if the set of points
${\mathcal S}$ is understood from the context.
\end{definition}

If ${\mathcal L} = {\mathcal O}_X(D)$
for an $\bFq$-rational divisor $D$ on $X$ whose support 
is disjoint from 
$\{P_1,\ldots,P_n\}$, then up to monomial equivalence,
this is the same as the evaluation code as in (\ref{ev}) from the subspace
${\mathcal F}$ of the field of rational functions of $X$
given by
\[
{\mathcal F} = \{f \in \bFq(X)^* :\, {\rm div}(f) + D \ge 0\} \cup \{0\}.
\]
For instance, when
$X$ is a smooth algebraic curve and ${\mathcal L} = {\mathcal O}_X(G)$
for some divisor $G$ defined over $\bFq$ whose support
is disjoint from the support of $D = P_1 + \cdots + P_n$,
then this is the same as the algebraic geometric
Goppa code $C_L(D,G)$ from $X$.

For explicit constructions of codes from embedded varieties
$X \subseteq \proj^m$, another more elementary description is
also available using \emph{homogeneous coordinates}
$(a_0 : a_1 : \cdots : a_m)$ for points in $\proj^m$,
where $(a_0 : a_1 : \cdots : a_m)$ and
$(\lambda a_0 : \lambda a_1 : \cdots : \lambda a_m)$
represent the same point whenever $\lambda \in \F^{\, *}$.

\begin{definition}\label{elem}
Choosing any one such homogeneous
coordinate vector defined over $\bFq$ for each of the points $P_i$ in
the set ${\mathcal S}$, define an evaluation map
$ev_{\mathcal S}$ and a code as in (\ref{ev}) using
the vector space ${\mathcal F}_1$ of linear forms (homogeneous
polynomials of degree 1) in $\bFq[x_0,\ldots,x_m]$.  The code
obtained as the image of this mapping is often denoted $C(X)$,
or $C(X; {\mathcal S})$ if it is important to specify the set of points.
Similarly, the space of linear forms can be replaced by the vector
space ${\mathcal F}_h$ of homogeneous polynomials of any degree $h \ge 1$, and
corresponding codes denoted $C_h(X; {\mathcal S})$ or $C_h(X)$
are obtained.
\end{definition}

\begin{example}\label{RM}
Let $X = \proj^m$ itself, and let ${\mathcal S}$
be the set of \emph{affine} $\bFq$-rational points of $X$,
that is, points in the complement of the hyperplane
${\bf V}(x_0)$, having homogeneous coordinate vectors of the
form $(1 : a_1 : \ldots : a_m)$.  With these particular
coordinate vectors, the code $C_h(X; {\mathcal S})$
is the well-known $q$-ary $h$th order (generalized)
\emph{Reed-Muller} code, denoted ${\mathcal R}_q(h,m)$.  (When $m = 1$,
this is the same as an \emph{extended Reed-Solomon} code.)
The block length
is $n = q^m$.  If $h < q$, then
the monomials $x^\beta = x_0^{\beta_0} \cdots x_m^{\beta_m}$
where $|\beta| = \beta_0 + \cdots + \beta_m = h$
are linearly independent on ${\mathcal S}$, so the dimension of
${\mathcal R}_q(s,m)$ is $k = \binom{m + h}{h}$.
If ${\mathcal S} = \proj^m(\bFq)$, the resulting
\emph{projective Reed-Muller codes} have block length
$n = q^m + \cdots + q + 1$.
$\diamondsuit$
\end{example}

There is, of course, a tight connection between
Definition~\ref{sophis} and Definition~\ref{elem}.  If $X$ is
embedded in $\proj^m$ and ${\mathcal L} = {\mathcal O}_X(1)$ is the
hyperplane section bundle, then $C(X,{\mathcal O}_X(1))$ and $C(X)$
are monomially equivalent codes (they differ at most by constant
multiples in each component depending on how the isomorphisms of
the fibers with $\bFq$ are chosen).  Similarly, $C_h(X)$ is
equivalent to $C(X,{\mathcal O}_X(h))$.
Also, in theory it suffices to consider the $C(X) = C_1(X)$
codes, since the $C_h(X)$ code on $X$ is the same as the $C_1$
code on the variety $\nu_h(X)$, where $\nu_h$ is the degree-$h$
\emph{Veronese mapping}
\begin{eqnarray*}
\nu_h : \proj^m & \longrightarrow & \proj^{\binom{m + h}{h}-1}\\
     (x_0 : x_1 : \cdots : x_m) &\mapsto & ( \cdots : x^\beta : \cdots ),
\end{eqnarray*}
and $x^\beta = x_0^{\beta_0} \cdots x_m^{\beta_m}$ ranges over all
monomials of total degree $h$.   The image $\nu_h(\proj^m)$ has
dimension $m$, degree $h^m$, and is isomorphic to $\proj^m$.

\section{Estimating the Parameters}\label{params}

\subsection{Elementary bounds}
Suppose Definition~\ref{elem} is used to construct a code
$C_h(X; {\mathcal S})$ from a variety $X$.  The block length
of the code is $n = \# {\mathcal S}$.  Using
a standard linear algebra result, the dimension is
\[
k = \dim {\mathcal F}_h - \dim \ker ev_{\mathcal S}.
\]
Forms of degree $h$ vanishing on $X$ always give 
elements of the kernel.  The dimension of the space of such forms
can be computed using the long exact cohomology sequence of 
\begin{equation}\label{sheaf}
0 \longrightarrow {\mathcal I}_X(h) \longrightarrow {\mathcal O}_{\proj^m}(h) 
\longrightarrow {\mathcal O}_{X}(h) \longrightarrow 0.
\end{equation}

Since each
codeword is $ev_{\mathcal S}(f) = (f(P_1),\ldots,f(P_n))$ for
some form $f$,
the codeword weight is $n - \# ({\bf V}(f) \cap {\mathcal S})$, 
the number of $P_i$ in ${\mathcal S}$ where $f$ is not zero.
Therefore,
\begin{equation}\label{mindist}
d = \min_{f \ne 0 \in {\mathcal F}_h} \left(n - \#({\bf V}(f) \cap {\mathcal S})\right).
\end{equation}

Along similarly general lines, let $\dim Y = \delta$ and let the degree of $Y$ be
$s < q + 1$ in $\proj^m$.  Let $E$ be an $\bFq$-rational linear
subspace of dimension $m - \delta - 1$ 
with $E \cap Y = \emptyset$.  By projection from $E$ onto a linear subspace 
$L \simeq \proj^\delta$, each $\bFq$-rational point of $L$ corresponds 
to at most $s$ such points of $Y$, so
\begin{equation}\label{linproj}
\# Y(\bFq) \le s\cdot \# \proj^\delta(\bFq) = s(q^\delta + \cdots + q + 1).
\end{equation}
Applying (\ref{linproj}) to $Y = X \cap H$ for a
hyperplane, Lachaud obtains the following elementary bound in \cite{l96}.

\begin{theorem}\label{elembound}
Let $X$ be a projective variety of dimension $\delta$ and
degree $s < q + 1$.  Then for $h = 1$ the $C(X)$ code has
\[
d \ge n - s(q^{\delta - 1} + \cdots + q + 1).
\]
\end{theorem}

A more refined estimate of the number of $\bFq$
rational points on a projective hypersurface establishes
the following result for the projective
Reed-Muller codes introduced in Example~\ref{RM}.

\begin{theorem}
\label{RMparams}
Let $h \le q$.
The projective Reed-Muller code of order $h$
has parameters
\[
\left[q^m + \cdots + q + 1, \binom{m + h}{h}, (q + 1 - h)q^{m-1}\right].
\]
\end{theorem}

\begin{proof}
Write ${\mathcal S} = \proj^m(\bFq)$.
The evaluation mapping is injective and $k = \dim {\mathcal F}_s = \binom{m + h}{h}$
provided that $d > 0$.  By \cite{se91}, if $f$ is a
homogeneous polynomial of degree $h \le q$, then (improving the bound
of (\ref{linproj}))
\[
\# ({\bf V}(f) \cap {\mathcal S}) \le  hq^{m-1} + q^{m-2} + \cdots + q + 1.
\]
Moreover, if ${\bf V}(f)$ is the union of $h$ $\bFq$-rational 
hyperplanes meeting along a common $(m - 2)$-dimensional linear 
subspace, this bound is attained.  Hence
\[
d = (q^m + q^{m-1} + \cdots + q + 1) - (hq^{m-1} + q^{m-2} + \cdots + q + 1)
= (q + 1 - h)q^{m-1}
\]
as claimed.
\end{proof}

In the remainder of this section, several other general techniques
for estimating the minimum distance of these codes
will be considered.  The first three are primarily geometric,
while the last is arithmetic in nature.

\subsection[]{Bounds from covering families of curves}

For the following discussion, it will be most convenient to
use the code construction given in Definition~\ref{sophis}.
In many concrete cases, it can be seen that
the points in the set ${\mathcal S}$ are distributed on a collection
of \emph{curves} $C_i$ (subvarieties of dimension $1$) on the variety $X$.
Since each section $f \in \Gamma(X,{\mathcal L})$ on $X$ defines
a divisor of zeroes $Z(f)$, a subvariety of codimension 1 on $X$,
determining the minimum distance of the $C(X,{\mathcal L})$
code reduces to understanding how many times the divisors $Z(f)$
can intersect the curves $C_i$ at points of ${\mathcal S}$.   To prepare,
let $C$ be any irreducible curve in $X$.
Observe that the divisors $Z(f)$ for $f \in \Gamma(X,{\mathcal L})$
all cut out divisors on $C$ of the same degree.  This degree
will be denoted by ${\mathcal L}\cdot C$.  In this situation, Hansen
derives a lower bound for $d$ in \cite{sh01}.

\begin{theorem}
\label{Hansencurves}
Let $X$ be a normal projective variety defined over $\bFq$,
of dimension $\dim X \ge 2$.  Let ${\mathcal S} \subseteq X(\bFq)$
and assume $\displaystyle{{\mathcal S} \subset \bigcup_{i=1}^a C_i}$
where $C_i$ are irreducible curves on $X$, also defined
over $\bFq$.  Assume that $\#(C_i \cap {\mathcal S}) \le N$ for all $i$.
Let ${\mathcal L}$ be a line bundle on $X$ defined
over $\bFq$ such that
\[
0 \le {\mathcal L}\cdot C_i \le \eta \le N
\]
for all $i$.  Let
\[
\ell = \max_{f\ne 0\in \Gamma(X,{\mathcal L})} \# \{i : Z(f)\ \text{contains}\ C_i\}.
\]
Then the code $C(X,{\mathcal L}; {\mathcal S})$ has
\[
d \ge \# {\mathcal S} - \ell N - (a - \ell) \eta.
\]
\end{theorem}

\begin{proof}
Let $f \in \Gamma(X,{\mathcal L})$, let $D = Z(f)$, and let
$E = Z(f) \cap \displaystyle{\bigcup_{i=1}^a C_i}$.
Suppose $E$ contains $\ell' \le \ell$ of the $C_i$.
The number points of ${\mathcal S}$ that are
contained in $E$ is estimated as follows:
\begin{eqnarray*}
\# (E\cap {\mathcal S}) &\le & \ell' N + (a - \ell') \eta\\
                 &\le & \ell N + (a - \ell)\eta
\end{eqnarray*}
(since by hypothesis $\eta \le N$).
Hence $ev_{\mathcal S}(f)$ has at least
$\# {\mathcal S} - \ell N - (a - \ell)\eta$ nonzero entries.
\end{proof}

\begin{example}\label{prod}
Let $X = \proj^1\times \proj^1$.  Let ${\mathcal S} = X(\bFq)$,
which consists of $(q+1)^2$ points, equally distributed over the
lines $C_1,\ldots,C_{q+1}$ of one of the rulings.
The Picard group of line bundles modulo isomorphism is
${\rm Pic}(X) \simeq \mathbb{Z} \oplus \mathbb{Z}$, so
the lines $C_i$ may be taken as the divisors of zeros of
sections of a line bundle of type $(1,0)$.
Let ${\mathcal L}$ have type $(\alpha,\beta)$ where $0\le \alpha,\beta \le q+1$.
Apply Theorem~\ref{Hansencurves} to estimate $d$
for the $C(X,{\mathcal L})$ code.  Because of the description of ${\mathcal S}$
above, $N = q+1$.
The divisor $Z(f)$ for $f\in \Gamma(X,{\mathcal L})$ contains at
most $\alpha$ of the $C_i$, so $\ell = \alpha$.
Moreover, ${\mathcal L}\cdot C_i = \beta$
for each $i$, so $\eta = \beta$.  The bound is
\[
d \ge (q + 1)^2 - \alpha (q + 1) - (q + 1 - \alpha)\beta =
(q + 1 - \alpha)(q + 1 - \beta).
\]
It is easy to construct codewords of this weight via bihomogeneous
polynomials on $\proj^1\times\proj^1$.  So this is the exact
minimum distance.  
$\diamondsuit$
\end{example}

\subsection{Bounds using Seshadri constants}

A second general method for estimating the
minimum distance of the $C(X, {\mathcal L}; {\mathcal S})$ codes
is based on the \emph{Seshadri constant} of ${\mathcal L}$
relative to the set ${\mathcal S}$.  This is potentially useful
but requires some significantly
more sophisticated birational geometry to state and apply.
Let $\pi : Y \to X$
be the blow up of the $X$ at the points in ${\mathcal S}$ and call
the exceptional divisor $E$.  Then
the Seshadri constant is defined as
\[
\varepsilon({\mathcal L},{\mathcal S}) = \sup \{\varepsilon \in \mathbb{Q} :
\pi^*{\mathcal L} - \varepsilon E \text{ is nef on } Y\}.
\]
(Here, ``nef'' means \emph{numerically effective}, that 
is, $(\pi^*{\mathcal L} - \varepsilon E)\cdot C \ge 0$ for all irreducible
curves $C$ on $Y$.)
Hansen proves the following estimate for the minimum distance of the
$C(X, {\mathcal L}; {\mathcal S})$ codes in \cite{sh01}.

\begin{theorem}\label{Sesh}
Let $X$ be a nonsingular projective variety of dimension $\ge 2$ over $\bFq$.
If ${\mathcal L}$ is ample with Seshadri constant
$\varepsilon({\mathcal L},{\mathcal S}) \ge e \in \mathbb{N}$,
and $n > e^{1-\dim(X)}{\mathcal L}^{\dim(X)}$, then
$C(X, {\mathcal L}; {\mathcal S})$ has minimum distance
$d \ge n - e^{1-\dim(X)} {\mathcal L}^{\dim(X)}$.
\end{theorem}

This is particularly well-suited for analyzing certain codes from
Deligne-Lusztig varieties to be defined in Section~\ref{DLvars} below.

\subsection{Bounds from ${\mathcal S}$ itself}

All of the $C_h(X; {\mathcal S})$ codes introduced in Section~\ref{constr}
can be viewed as punctures of the projective Reed-Muller code
of order $h$ on the appropriate $\proj^m$ (delete the components
corresponding to points in the complement of ${\mathcal S}$).  For 
this reason, in addition to making use of the properties of the 
variety $X$, it is also possible to use properties of the $0$-dimensional
algebraic set (or scheme) ${\mathcal S}$ itself to estimate $d$.
Let ${\mathcal I}_{\mathcal P}$ be the sheaf of ideals defining any $0$-dimensional
${\mathcal P}$.  From the long exact cohomology sequence of the
exact sequence of sheaves
\[
0 \longrightarrow {\mathcal I}_{\mathcal P} \longrightarrow {\mathcal O}_{\proj^m} \longrightarrow {\mathcal O}_{\mathcal P} \longrightarrow 0,
\]
it follows that for all $h \ge 0$,
\begin{equation}\label{lecs}
0 \to H^0({\mathcal I}_{\mathcal P}(h)) \to H^0({\mathcal O}_{\proj^m}(h)) \to H^0({\mathcal O}_{\mathcal P}(h)) \to H^1({\mathcal I}_{\mathcal P}(h)) \to 0.
\end{equation}
The term $H^0({\mathcal I}_{\mathcal P}(h))$ gives the space of homogeneous forms of degree
$h$ vanishing on ${\mathcal P}$.  The term $H^1({\mathcal I}_{\mathcal P}(h))$ measures the failure
of the points in ${\mathcal P}$ to impose independent conditions on forms of degree $h$.  

In the case that ${\mathcal S}$ is a \emph{complete intersection}
of hypersurfaces of degrees $d_1,\ldots,d_m$, there
are particularly nice
techniques from commutative algebra and algebraic geometry
related to the classical \emph{Cayley-Bacharach Theorem} that apply.
A modern version of this result due to Davis, Geramita,
and Orecchia can be stated as follows in the
situation at hand.

\begin{theorem}\label{CB}
Let ${\mathcal S} \subset \proj^m$ be a reduced complete intersection of
hypersurfaces of degrees $d_1,\ldots,d_m$.  Let $\Gamma', \Gamma''$
be disjoint subsets of ${\mathcal S}$ with ${\mathcal S} = \Gamma' \cup \Gamma''$.
Let $s = \sum_{i=1}^m d_i - m - 1$.  Then for all $h \ge 0$,
\[
\dim H^0({\mathcal I}_{\Gamma'}(h)) - \dim H^0({\mathcal I}_{\mathcal S}(h)) = \dim H^1({\mathcal I}_{\Gamma''}(s - h)).
\]
\end{theorem}

Hence, one way to interpret Theorem~\ref{CB} is that when $\Gamma' \subset {\mathcal S}$,
the difference in dimension between the space of homogeneous forms of
degree $a$ vanishing on $\Gamma'$ and the subspace vanishing
on ${\mathcal S}$ is equal to the dimension of $H^1({\mathcal I}_{\Gamma''}(s - a))$.
Moreover by (\ref{lecs}), this dimension measures the failure of
$\Gamma''$ to impose independent conditions on homogeneous forms of
degree $s - a$.  

Applied to the corresponding codes from ${\mathcal S}$
consisting of $d_1d_2\ldots d_m$ distinct $\bFq$-rational
points, this result implies the
following.

\begin{theorem}\label{CBbd}
Let ${\mathcal S}$ be a reduced complete intersection of hypersurfaces
of degrees $d_1,\ldots,d_m$ in $\proj^m$.  Let $s = \sum_{i=1}^m d_i - m - 1$
as in Theorem~\ref{CB}.  If $1 \le h\le s$, the code $C_h({\mathcal S})$
has minimum distance
\[
d \ge \sum_{i=1}^m d_i - h - (m - 1) = s - h + 2.
\]
\end{theorem}

The proof is accomplished by showing that under these hypotheses,
any form of degree $h$ that is zero on a subset $\Gamma'$ that is
too large must be zero at all points in ${\mathcal S}$ because the
$H^1({\mathcal I}_{\Gamma''}(s - h))$ group vanishes.

The bound on $d$ given here was improved rather strikingly 
by Ballico and Fontanari to $d \ge m(s - h) + 2$ 
under the assumption that all subsets of $m+1$ of the points in ${\mathcal S}$
span $\proj^m$ -- see \cite{bf06} for this. 

Bounds derived by these methods
are usually interesting only for $h$ close to $s$.  
Moreover some, but not all, interesting examples of ${\mathcal S}$ satisfy
the complete intersection hypothesis.  For instance the affine
$\bFq$-rational points in $\proj^m$ form a complete intersection
for all $m$.  The $\F_8$-rational
points on the Klein quartic and the $\F_{r^2}$ points on the
Hermitian curve are other examples.    
  
\subsection{General Weil-type bounds}

From (\ref{mindist}) above, and the proof of
Theorem~\ref{RMparams}, the minimum distance
of a $C(X)$ code as in Definition~\ref{elem} is determined by
the numbers of $\bFq$-rational points on the subvarieties
$Y = X \cap {\bf V}(f)$.  Hence, another possible approach to estimate $d$
is to apply general bounds for $\# Y(\bFq)$, for instance
bounds derived from the statements of the Weil conjectures,
or refined versions of these.

We very briefly recall the deep mathematics behind this approach.
Thinking of $X$ as a variety over the algebraic closure of the
finite field, the number
of $\bFq$-rational points on $X$ can be computed
by an analog of the Lefschetz trace formula for the action of
the Frobenius endomorphism $F$ on the $\ell$-adic \'etale
cohomology groups of $X$, $H^i(X)$ (where $\ell$ is any prime not
dividing $q$):
\begin{equation}
\label{lef}
\# X(\bFq) = \sum_{i=0}^{2m} (-1)^i\, {\rm Tr}(F | H^i(X)).
\end{equation}
Moreover, the eigenvalues of $F$ on $H^i(X)$ are algebraic
numbers of absolute value $q^{i/2}$.  When $X$ is obtained
from a variety $Y$ defined over the ring of integers $R$ of some
number field by reduction modulo some prime ideal in $R$, then
the dimensions of the $H^i(X)$ are the same as the topological
Betti numbers of the variety over $\mathbb{C}$ corresponding
to $Y$.

Thus, for instance, if $X$ is a smooth curve of genus
$g$ which is the reduction of a smooth curve $Y$, then
\[
\# X(\bFq) = 1 + q - \sum_{j=0}^{2g} \alpha_j,
\]
where $|\alpha_j| = q^{1/2}$ for all $j$.  The Hasse-Weil bound
often used in the theory of Goppa codes from curves
is a direct consequence:
\[
|\#X(\bFq) - (1 + q)| \le 2g\sqrt{q}.
\]
There is a correspondingly concrete Weil-type bound
for hypersurfaces in $\proj^m$, and this can be used to
derive bounds on the numbers of $\bFq$-rational points
in hyperplane sections as well.  A hypersurface is said
to be \emph{nondegenerate} if it not
contained in any linear subspace of $\proj^m$.

\begin{theorem} \label{hypWeil}
Let $X$ be a smooth nondegenerate hypersurface of degree $s$ in
$\proj^m$, $m \ge 2$. Then
\begin{equation}\label{hyp}
|\# X(\bFq) - (q^{m-1} + \cdots + q + 1)| \le b(s) q^{(m-1)/2},
\end{equation}
where $b(s) = \frac{s - 1}{s} ((s-1)^m - (-1)^m)$
is the middle Betti number of a smooth hypersurface of degree $s$
when $m$ is even, and one less than that number when $m$ is odd.
\end{theorem}

The inequality (\ref{hyp})
follows from the shape of the cohomology groups
$H^i(X)$ of a smooth hypersurface in $\proj^m$, which (by
the Lefschetz hyperplane theorem and Poincar\'e
duality) look like the corresponding groups for $\proj^{m-1}$,
except possibly in the middle dimension $i = m - 1$.

\begin{example}\label{Curve}
If $m = 2$ and $X$ is a smooth curve of degree $s$ in $\proj^2$, then
\[
b(s) = \frac{s-1}{s}((s-1)^2 - 1) = (s-1)(s-2) = 2g(X)
\]
as expected.  In order to obtain long codes over $\bFq$,
the \emph{maximal curves}, that is,
curves attaining the maximum $\# X(\bFq)$ from (\ref{hyp}),
have been especially intensively studied.  For instance, when
$q = r^2$, the Hermitian curve of degree $s = r+1$ over $\F_{r^2}$,
$X = {\bf V}(x_0^{r+1} + x_1^{r+1} + x_2^{r+1})$,
has $\# X(\F_{r^2}) = r^3 + 1 = 1 + r^2 + r(r - 1)r$.
$\diamondsuit$
\end{example}

\begin{example}\label{Hsurf}
When $m = 3$ and $q = r^2$, the analogous \emph{Hermitian
surfaces} $X = {\bf V}(x_0^{r+1} + x_1^{r+1} + x_2^{r+1} + x_3^{r+1})$
also attain the upper bound from (\ref{hyp}), which reads
\[
\# X(\F_{r^2}) \le 1 + r^2 + r^4 + \frac{r}{r+1}(r^3+1) r^2 =
(r^2 + 1)(r^3 + 1).
\]
The Hermitian surface contains this many distinct
$\F_{r^2}$-rational points because, for instance,
it is possible to take the defining equation to the affine
form
\[
y_1^r + y_1 = y_2^{r+1} + y_3^{r+1}
\]
by a linear change of coordinates that puts a plane tangent
to the surface as the plane at infinity.  Then there are
$r^5$ affine $\F_{r^2}$-rational points ($r$ for each
pair $(y_2,y_3) \in (\F_{r^2})^2$).  There are also
$(r+1)r^2 + 1$ rational points at infinity since the intersection of
the surface with each of its tangent planes at an $\F_{r^2}$-rational
point is the union of
$r+1$ concurrent lines in that plane.  This yields
$r^5 + (r+1)r^2 + 1 = (r^3+1)(r^2+1)$ points as claimed.
$\diamondsuit$
\end{example}

The following result of Lachaud appears in \cite{l96}.

\begin{theorem}\label{hypsecbds}
Let $X$ be a smooth nondegenerate hypersurface of degree $s$ in $\proj^m$
for $m \ge 3$.  Let $H = {\bf V}(f)$ for a linear form
in $\bFq[x_0,\ldots,x_m]$, and let $X_H$ denote the intersection
$X \cap H$ (with the reduced scheme structure).  Then
\begin{equation}\label{hypsec}
|\# X_H(\bFq) - (q^{m-2} + \cdots + q + 1)| \le (s-1)^{m-1} q^{(m-1)/2},
\end{equation}
and
\begin{equation}\label{hypcomp}
|q \# X_H(\bFq) - \# X(\bFq) | \le (s - 1)^{m-1} (q + s - 1)q^{(m-1)/2}.
\end{equation}
\end{theorem}

These bounds are proved by comparing the cohomology of $X$ and $X_H$, taking
into account possible singularities of $X_H$.
For a proof, see Corollary 4.6 and preceding results of \cite{l96}.

When ${\mathcal S}$ is the full set of $\bFq$-rational
points on $X$, so $n = \# {\mathcal S}$ for the
$C(X; {\mathcal S})$ code and $H$ is a general hyperplane,
these imply the following bounds
on $\# (H \cap {\mathcal S})$.
(\ref{hypcomp}) implies
\begin{equation}\label{genbd1}
\left|(n - \#(H\cap {\mathcal S})) - \frac{(q-1)}{q} n\right|
\le (s - 1)^{m-1} (q + s - 1)q^{(m-1)/2}
\end{equation}
and
\begin{equation}\label{genbd2}
\left| (n - \#(H\cap {\mathcal S}))- q^{m-2}\right|
\le s(s-1)^{m-1}q^{(m-1)/2}.
\end{equation}
These, together with (\ref{hypsec}), give universally applicable
lower bounds on $d$ by applying (\ref{mindist}).

As is perhaps to be expected, it is often possible to derive tighter
bounds in specific cases by taking
the properties of $X$ into account.

\section{Examples}\label{exs}

This section will consider codes produced
according to the constructions from Section~\ref{constr} from
various special classes of varieties.  The particular
varieties used here are all examples of varieties
with many rational points over finite fields $\bFq$.
The examples are ordered according to the algebraic
geometric prerequisites needed for the construction.

\subsection{Quadrics}

First consider the $C(X)$
codes from quadric hypersurfaces $X = {\bf V}(f)$
for homogeneous $f$ of degree 2 in $\bFq[x_0,\ldots,x_m]$.
The following statements are proved, for instance, in Chapter 22 of
\cite{ht91}.
Up to projective equivalence over $\bFq$,
such $X$ are completely described by a positive integer
called the \emph{rank} and a second integer called
the \emph{character}, which takes values in the finite
set $\{0,1,2\}$.
The rank, denoted $\rho$, can be described as
the minimum number of variables needed to express $f$ after
a linear change of coordinates in $\proj^m$.  $X$ is
said to be \emph{nondegenerate} if $\rho = m+1$.  Nondegenerate
quadrics are always smooth varieties.
Degenerate quadrics are singular, but they are cones over nondegenerate
quadrics in a linear subspace of $\proj^m$.  Hence in principle
it suffices to study nondegenerate quadrics and we will consider
only that case here.   The character, denoted $w$, is most easily
described by considering a finite set of possible
normal forms for $f$.

If $m$ is even, then every nondegenerate quadric can be
taken to the form
\[
x_0^2 + x_1x_2 + x_3x_4 + \cdots + x_{m-1}x_m.
\]
${\bf V}(f)$ is called a \emph{parabolic} quadric in this case,
and the character $w$ is defined to be $1$.

On the other hand, if $m$ is odd, there are two distinct possible forms:
\begin{eqnarray*}
& x_0 x_1 + x_2x_3 + \cdots + x_{m-1} x_m & \text{or}\\
& q(x_0,x_1) + x_2 x_3 + \cdots + x_{m-1}x_m.
\end{eqnarray*}
In the first case, ${\bf V}(f)$ is called a \emph{hyperbolic}
quadric and $w = 2$.  In the second, $q(x_0,x_1)$ is a quadratic
form in two variables which can be further reduced to slightly
different normal forms depending on whether $q$ is even or odd.
For both even and odd $q$, in the second case, ${\bf V}(f)$ is
called a \emph{elliptic} quadric and $w = 0$.

\begin{theorem}\label{quad}
Let $X$ be a nondegenerate quadric in $\proj^m$ with
character $w$.  Then
\begin{eqnarray*}
\# X(\bFq) &=& \frac{(q^{(m+1-w)/2} + 1)(q^{(m - 1 + w)/2} - 1)}{q-1}\\
&=& q^{m-1} + \cdots + q + 1 + (w-1)q^{(m-1)/2}.
\end{eqnarray*}
\end{theorem}

In particular, this result says that hyperbolic and parabolic quadrics
attain the upper bound from (\ref{hyp}) with $s = 2$,
and elliptic quadrics attain the lower bound.

Because each linear section of $X$ is also a quadric
in a lower-dimensional space, Theorem~\ref{quad} can be used
to determine the full weight distributions of the $C(X)$ codes.
In particular,

\begin{theorem}\label{quadth}
The $C(X)$ code from a smooth quadric $X$ in $\proj^m$ has
$n$ given in Theorem~\ref{quad}, $k = m+1$ and 
\begin{equation}\label{quadd}
d = \begin{cases}
    q^{m-1} & \text{if\ } w = 2\\
    q^{m-1} - q^{(m-2)/2} & \text{if\ } w = 1\\
    q^{m-1} - q^{(m-1)/2} & \text{if\ } w = 0.\\
    \end{cases}
\end{equation}
\end{theorem}

For instance, if $m$ is even, so $w = 1$ (the parabolic case), the
hyperplane section of $X$ containing the most $\bFq$-rational points
will be a hyperbolic section and $d$ is as above.  When $w = 2$
(for example, for codes from hyperbolic quadrics in $\proj^3$),
the minimum weight codewords come from hyperplane sections that are
degenerate quadrics.  

The same sort of reasoning has also be used by Nogin and Wan
to determine the
complete hierarchy of generalized Hamming weights
$d_1(C(X)) , \ldots , d_k(C(X))$.
  The results are somewhat intricate
to state, though, so we refer the interested reader to the articles
\cite{wan93,n93} and the notes in Section~\ref{bibnotes}.

For the $C_h(X)$ codes with $h \ge 2$, the dimension can be
estimated using (\ref{sheaf}), where ${\mathcal I}_X(h) \simeq {\mathcal O}_{\proj^m}(h-2)$.
This yields
\[
k \le \binom{m + h}{h} - \binom{m + h - 2}{h - 2}.
\]

\subsection{Hermitian hypersurfaces}

For the $C(X)$ codes constructed from the Hermitian surfaces
of Example~\ref{Hsurf} with $q = r^2$, (\ref{hypsec}) gives
\[
d \ge (r^2 + 1)(r^3 + 1) - (r^2 + 1 + r^4) = r^5 - r^4 + r^3.
\]
However, closer examination of the hyperplane sections of the
Hermitian surface yields the following statement.

\begin{theorem}\label{hermsurfbound}
Let
$X = {\bf V}(x_0^{r+1} + x_1^{r+1} + x_2^{r+1} + x_3^{r+1})$
be the Hermitian surface over $\F_{r^2}$.  The $C(X)$ code
on ${\mathcal S} = X(\F_{r^2})$ has parameters
\[
\left[(r^2+1)(r^3+1), 4, r^5 \right].
\]
\end{theorem}

\begin{proof}
Every $\F_{r^2}$-rational plane in $\proj^3$
intersects $X$ either in a Hermitian curve containing $r^3 + 1$
points over $\F_{r^2}$, or else in $r+1$ concurrent lines
containing $(r+1)r^2 + 1$ points.  Hence by (\ref{mindist}),
\[
d = n - ((r+1)r^2+1) = r^5.
\]
\end{proof}

The $C_h(X)$ codes with $h > 1$ are more subtle here.

\begin{theorem} Let $X$ and ${\mathcal S}$ be as in Theorem~\ref{hermsurfbound}.
If $h < r + 1$, the $C_h(X)$ code has parameters
\[
\left[(r^2+1)(r^3+1), \binom{4 + h}{h}, d \ge n - h(r+1)(r^2+1) \right].
\]
\end{theorem}

\begin{proof}
This bound follows from Theorem~\ref{elembound}
by the fact that if $f$ is a form of degree
$h$, then ${\bf V}(f) \cap X$ is a curve of degree $\delta = h(r+1)$
in $\proj^3$. The hypothesis on $h$ implies that the evaluation mapping is injective.
For larger $h$, (\ref{sheaf}) would be used to determine the dimension
of the space of forms of degree $h$ vanishing on the Hermitian variety.
\end{proof}

An even tighter bound
\begin{equation}\label{Sorconj}
d \ge n - (h(r^3 + r^2 - r) + r + 1)
\end{equation}
has been conjectured by S\o rensen for these codes in \cite{so91a}.

The Hermitian curve and surface codes can be generalized as follows.
(see Chapter 23 of\cite{ht91}).  Over a field
of order $q = r^2$, consider the Hermitian hypersurface in $\proj^m$
defined by
\begin{equation}\label{hermhyp}
X = {\bf V}(x_0^{r + 1} + x_1^{r + 1} + \cdots + x_m^{r + 1}).
\end{equation}
The mapping $F(x) = x^r$ is a involutory field automorphism of $\F_{r^2}$, analogous
to complex conjugation in $\C$, and the homogeneous polynomial
defining $X$ is analogous to the usual Hermitian form on $\C^{m+1}$
given by $x_0 \overline{x_0} + \cdots + x_m \overline{x_m}$.
The defining polynomial of $X$ may be understood as $H(x,x)$ for the
mapping $H : \F_{r^2}^{\,m+1} \times \F_{r^2}^{\,m+1} \to \F_{r^2}$
given by
\[
H(x,y) = x_0 y_0^r + \cdots + x_m y_m^r.
\]
It is clear that $H$ is additive in each variable and satisfies
$H(\lambda x,y) = \lambda H(x,y)$ and
$H(x,\lambda y) = \lambda^r H(x,y) = F(\lambda)H(x,y)$
for the automorphism $F$ above.   Hence $H$ is an example of what
is known as a \emph{sesquilinear form} on $\F_{r^2}^{\,m + 1} \times \F_{r^2}^{\,m + 1}$.
It can be shown that after a linear change of coordinates defined
over $\F_{r^2}$,
any sesquilinear $H$ on $V \times V$, where $V$ is a finite-dimensional
$\F_{r^2}$-vector space, can be expressed as
\begin{equation}\label{hermform}
H(x,y) = x_0 y_0^r + \cdots + x_\ell y_\ell^r
\end{equation}
for some $\ell \le \dim V$.
$H$ is said to be nondegenerate if $\ell = \dim V$ and
degenerate otherwise.

It follows that every linear section $L\cap X$ of a Hermitian hypersurface
is also a Hermitian variety in the linear subspace $L = \proj W$ for
some vector subspace $W$.  Moreover,
if the section is degenerate (i.e. $\ell < \dim W$ in (\ref{hermform})),
then the section is a cone over a nondegenerate
Hermitian variety in a linear subspace of $L$.  Thus,
the properties of the codes $C(X)$ from the
Hermitian hypersurfaces are formally quite similar to (and even
somewhat simpler than) the properties of codes from quadrics
discussed above.  The main ingredient is the following
statement for the nondegenerate Hermitian hypersurfaces.

\begin{theorem}\label{hypcount}
Let $X$ be the nondegenerate Hermitian hypersurface from (\ref{hermhyp}).
Then
\[
\# X(\F_{r^2}) = r^{2m-2} + \cdots + r^2 + 1 + b(r+1) r^{m-1},
\]
where $b(r+1) = \frac{r}{r+1} (r^m - (-1)^m)$.
\end{theorem}

In other words, for all $m$, the nondegenerate
Hermitian hypersurfaces meet the upper bound
from (\ref{hyp}) for a hypersurface of degree $s = r + 1$.

\begin{theorem}\label{hermhypcodes}
Let ${\mathcal S} = X(\F_{r^2})$ for the nondegenerate Hermitian hypersurface
$X$ in $\proj^m$.  The $C(X;{\mathcal S})$ code has $n$ given in 
Theorem~\ref{hypcount}, $k = m + 1$, and
\[
d = \begin{cases} 
     r^{2m-1} - r^{m - 1} & \text{\ if\ } m \equiv 0 \bmod 2\\
     r^{2m-1} & \text{\ if\ } m \equiv 1 \bmod 2.
    \end{cases}
\]
\end{theorem}

When $m$ is even, the minimum weight codewords of the $C(X)$
come from nondegenerate Hermitian variety hyperplane sections.
On the other hand, if $m$ is odd, then the minimum weight codewords
of $C(X)$ come from hyperplane sections that are degenerate
Hermitian varieties.  In this case,
In both cases, the nonzero codewords of $C(X)$ have only two distinct
weights:
\[
r^{2m - 1} + (-1)^{m-1} r^{m - 1} \text{ and } r^{2m - 1}.
\]
The hierarchies of generalized Hamming weights $d_r$ are also
known for the $C(X)$ codes by work of Hirschfeld, Tsfasman, and Vladut,
\cite{htv94}.
The same sort of techniques used in
Theorem~\ref{hermsurfbound} above can be applied to the
$C_h(X)$ codes for $h \ge 2$ here.  However, much less is known about
the exact Hamming weights of these codes.

\subsection{Grassmannians and flag varieties}

The \emph{Grassmannian} $\G(\ell,m)$ is a projective variety
whose points are in one-to-one correspondence with
the $\ell$-dimensional vector subspaces of an
$m$-dimensional vector space (or equivalently
the ($\ell-1$)-dimensional linear subspaces of $\proj^{m-1}$).
We very briefly recall the construction.

Let $\F$ denote an algebraic closure of $\bFq$.
Given any basis $B = \{v_1,\ldots,v_\ell\}$ for an $\ell$-dimensional
vector subspace $W$ of $\F^{\,m}$, form the
$\ell\times m$ matrix  $M(B)$ with rows $v_i$.  Consider the
determinants of the maximal square ($\ell\times \ell$)
submatrices of $M(B)$.   There is one such maximal
minor for each subset $I \subset \{1,\ldots,m\}$ with $\# I = \ell$,
so writing $p_I(W)$ for the maximal minor in the columns corresponding to $I$,
the \emph{Pl\"ucker coordinate vector} of $W$ is the homogeneous
coordinate vector
\begin{equation}\label{pluck}
p(W) = ( \cdots : p_I(W) : \cdots ) \in \proj^{\binom{m}{\ell} - 1},
\end{equation}
where $I$ runs through all subsets of size $\ell$ in $\{1,\ldots,m\}$.
The point $p(W)$ is a well-defined invariant of $W$
because a change of basis in $W$ multiplies the matrix
$M(B)$ on the left by the change of basis matrix, an element of
${\rm GL}(\ell,\F)$. All
components of the Pl\"ucker coordinate vector are multiplied
by the determinant of the change of basis matrix, an element of
$\F^{\,*}$.   Hence any choice of basis in $W$
yields the same point $p(W)$ in $\proj^{\binom{m}{\ell} - 1}$.

The locus of all such points (for all $W$) forms the
Grassmannian $\G(\ell,k)$.  Consider the set of $W$ such
that $p_{I_0}(W) \ne 0$, so
the maximal minor with $I_0 = \{1,\ldots,\ell\}$ is invertible.
The set of such $W$ is one of the
open subsets in the standard affine cover of $\G(\ell,m)$.  In
the row-reduced echelon form of $M(B)$, the entries in the columns
complementary to $I_0$ (an $\ell \times (m - \ell)$ block)
are arbitrary and uniquely determine $W$.  Hence
\[
\dim \G(\ell,m) = \ell (m - \ell).
\]
To construct Grassmannian codes, one uses the $\bFq$-rational
points of $\G(\ell,m)$, which come from subspaces $W$
defined over $\bFq$.  Nogin has established the following
result.

\begin{theorem}\label{Gcodes}
Let ${\mathcal S}$ be the set
of all the $\bFq$-rational points on $X = \G(\ell,m)$.
Then the $C(X; {\mathcal S})$ code (from linear forms in the
Pl\"ucker coordinates) has parameters
\[
\left[ \left[\begin{matrix} m\\ \ell\\ \end{matrix}\right]_q,\
\binom{m}{\ell},\ q^{\ell(m - \ell)}\right],
\]
where
\[
\left[\begin{matrix} m\\ \ell\\ \end{matrix}\right]_q =
\frac{(q^m - 1)(q^m - q)\cdots (q^m - q^{\ell-1})}{(q^\ell-1)(q^\ell-q)\cdots (q^\ell - q^{\ell-1})}.
\]
\end{theorem}

\begin{proof}
The numerator in the formula for
$\left[\begin{matrix} m\\ \ell\\ \end{matrix}\right]_q$
is precisely the number of ways of picking a list of $\ell$ linearly independent
vectors in $\bFq^{\,m}$ (a basis for a $W$ defined over $\bFq$).
Similarly, the denominator is the number of ways
of picking $\ell$ linearly independent vectors in $\bFq^{\,\ell}$, hence the order of the group ${\rm GL}(\ell,\bFq)$.  The quotient is the number of distinct
$\ell$-dimensional subspaces of $\bFq^{\,m}$.
This shows $n = \#{\mathcal S} = \left[\begin{matrix} m\\ \ell\\ \end{matrix}\right]_q$.  Assuming $d = q^{\ell(m-\ell)}$ for the moment,
the fact that $d > 0$ says the evaluation mapping on the vector space of
linear forms in $\proj^{\binom{m}{\ell} - 1}$ is injective, and the formula
for $k$ follows.  Finally, we must prove that $d = q^{\ell(m - \ell)}$.

The complement of the hyperplane section
$\G(\ell,m) \cap {\bf V}(p_{I_0})$ contains exactly $q^{\ell(m - \ell)}$
$\bFq$-rational points of $\G(\ell,m)$.
Hence $d \le q^{\ell(m - \ell)}$.
The cleanest way to prove that this is an equality
is to use the language of exterior
algebra on $\bFq$-vector spaces, following Nogin in \cite{n96}.

Let $V = \bFq^{\,m}$ and write
$e_i$ for the standard basis vectors in $V$.
The $\bFq$-rational points of the Grassmannian $\G(\ell,m)$ can be
identified with the subset of
$\proj\left(\bigwedge^\ell V\right)\simeq \proj^{\binom{m}{\ell} - 1}$
corresponding to the \emph{completely decomposable} elements of the exterior
product $\bigwedge^\ell V$ (that is, nonzero elements of the form
$\omega = w_1 \wedge w_2 \wedge \cdots \wedge w_\ell$ for some $w_i \in V$
that form a basis for the subspace they span).

The hyperplanes in $\proj\left(\bigwedge^\ell V\right)$ correspond to elements of
$\proj \left(\bigwedge^\ell V\right)^*$, hence
to elements of $\bigwedge^{m - \ell} V$ (up to scalars) via the nondegenerate pairing
\[
\textstyle{\wedge : \bigwedge^{m - \ell} V \times \bigwedge^\ell V \to \bigwedge^m V \simeq\bFq.}
\]
It follows that
the hyperplanes in $\proj\left(\bigwedge^\ell V\right)$ all have the form
\[
\textstyle H(\alpha) = \proj\, \{\omega \in \bigwedge^\ell V : \alpha \wedge \omega = 0\}
\]
for some nonzero $\alpha \in \bigwedge^{m - \ell} V$.

Under these identifications, each hyperplane ${\bf V}(f)$ for $f$ a
linear form in the Pl\"ucker coordinates corresponds to $H(\alpha)$ for some $\alpha$.
For instance, ${\bf V}(p_{I_0})$ corresponds
to $H(\alpha_0)$ for the completely decomposable
element $\alpha_0 = e_{\ell + 1} \wedge \cdots \wedge e_m$.
All completely decomposable $\alpha \in \bigwedge^{m - \ell} V$
define hyperplane sections of the  Grassmannian with the same number
of $\bFq$-rational points.  Call this number $N_\ell$.

What must be proved is
that if $\beta \in \bigwedge^{m - \ell} V$ is arbitrary, then the
linear forms $f$ in the Pl\"ucker coordinates defining the hyperplane
$H(\beta)$ satisfy
\[
{\rm wt}(ev_{\mathcal S}(f)) \ge N_\ell.
\]
This follows by induction on $\ell$ using the easily checked fact that
if $e \in V$ and $\alpha \in \bigwedge^{m - \ell} V$, then
\begin{equation}\label{anne}
\alpha \wedge e = 0 \Longleftrightarrow \alpha = \alpha' \wedge e
\end{equation}
for some $\alpha' \in \bigwedge^{m - \ell - 1} V$.

If $\ell = 1$,
there is nothing to prove because every element of $\bigwedge^{m - 1} V$
is completely decomposable.  If $\ell > 1$, writing $[\ell]_q = \# {\rm GL}(\ell,\bFq)$,
\begin{eqnarray*}
{\rm wt}(ev_{\mathcal S}(f))
& = & \# \{W = {\rm Span}(w_1,\ldots,w_\ell) : \beta \wedge w_1 \wedge \cdots \wedge w_\ell \ne 0\}\\
&=& \# \{(w_1,\ldots,w_\ell) : \beta \wedge w_1 \wedge \cdots \wedge w_\ell \ne 0\}\, /\, [\ell]_q
\end{eqnarray*}
Hence by the induction hypothesis, if $\alpha$ is completely decomposable
\begin{eqnarray*}
[\ell]_q\cdot {\rm wt}(ev_{\mathcal S}(f))
& = & \sum_{w_1 : \beta \wedge w_1 \ne 0} \# \{(w_2,\ldots,w_\ell) :
(\beta \wedge w_1) \wedge w_2 \wedge \cdots \wedge w_\ell \ne 0\}\\
& \ge & \sum_{w_1 : \beta \wedge w_1 \ne 0} N_{\ell - 1}\cdot [\ell - 1]_q \\
& = & N_{\ell - 1}\cdot [\ell - 1]_q\cdot \# \{w_1 : \beta \wedge w_1 \ne 0\}\\
&\ge & N_{\ell - 1}\cdot [\ell - 1]_q \cdot \# \{w_1 : \alpha \wedge w_1 \ne 0\} \quad \text{by\ (\ref{anne})}\\
& = & [\ell]_q\cdot N_\ell.
\end{eqnarray*}
\end{proof}

The exterior algebra language can also be used to say more about
the weight distribution of $C(\G(\ell,m); {\mathcal S})$.  For instance,
the number of minimum weight words of this code is equal to the number
of linear forms corresponding to completely decomposable $\alpha$.
This number is exactly $q - 1$ times the number of $\bFq$-rational points
of the dual Grassmannian $\G(m - \ell, m)$, or
\[
(q - 1) \left[\begin{matrix} m\\ m - \ell\\ \end{matrix}\right]_q
= (q - 1) \left[\begin{matrix} m\\ \ell\\ \end{matrix}\right]_q.
\]
For further information on these codes see the bibliographic
notes in Section~\ref{bibnotes}.

Codes on certain subvarieties
of Grassmannians, the so-called \emph{Schubert varieties},
have also been studied in detail by Chen, Guerra and Vincenti,
and Ghorpade and Tsfasman.
Let $\alpha = (\alpha_1,\ldots,\alpha_\ell)\in \mathbb{Z}^\ell$,
where $1 \le \alpha_1 \le \cdots \le \alpha_\ell \le m$. If
$B = \{v_1,\ldots,v_m\}$ is a fixed basis of $\bFq^{\,m}$, let
$A_i$ be the span of the first $i$ vectors in $B$. Then the
Schubert variety $\Omega_\alpha$ is defined as
\begin{equation}\label{schub}
\Omega_\alpha = \{p(W) \in \G(\ell,m) : \dim W \cap A_{\alpha_i} \ge i \}.
\end{equation}
See Section~\ref{bibnotes} for some pointers to
the literature here.

Just as Grassmannians parametrize linear subspaces
in $\F^{\,m}$, the \emph{flag varieties} parametrize flags of
linear subspaces, that is nested sequences of subspaces
\[
V_1 \subset V_2 \subset \cdots \subset V_s,
\]
where $\dim V_{i} = \ell_i$ and $0 < \ell_1 < \ell_2 < \ldots < \ell_s < m$.
The flag is said to have \emph{type} $(\ell_1,\ell_2,\ldots,\ell_s)$.
Also set $\ell_{s+1} = m$ and $\ell_0 = 0$ by convention.
The group $G = {\rm GL}(m,\F)$ acts on the set of flags of
each fixed type and
the isotropy subgroup of a particular flag is a parabolic subgroup $P$
conjugate to the group of block upper-triangular matrices
with diagonal blocks $M_r$ of sizes $\ell_r - \ell_{r - 1}$ for $1 \le r \le s+1$.
Hence the quotient $G/P$, which is denoted
${\mathcal F}(\ell_1,\ell_2,\ldots,\ell_s;m),$
classifies flags of type $(\ell_1,\ell_2,\ldots,\ell_s)$.
The set $G/P$ has the structure of a projective variety, which can be
described as follows.  Each $V_i$ corresponds to a point of $\G(\ell_i,m)$.
So the flag corresponds to a point of the product variety
$\G(\ell_1,m)\times \cdots \times \G(\ell_s,m)$ and
${\mathcal F}(\ell_1,\ell_2,\ldots,\ell_s;m)$ is the subset of this product
defined by the conditions $V_i \subset V_{i+1}$ for all $i$.   This can be
embedded in $\proj^{N_1} \times \cdots \times \proj^{N_s}$, for $N_i = \binom{m}{\ell_i}$,
by the Pl\"ucker coordinates as in (\ref{pluck}).  Finally, the product
\[
\proj^{N_1} \times \cdots \times \proj^{N_s} \hookrightarrow \proj^N
\]
for $N = (N_1 + 1)\cdot \cdots \cdot (N_s + 1) - 1$ by another
standard construction called the \emph{Segre map}.

As in the Grassmannian case, $\bFq$-rational points on 
${\mathcal F}(\ell_1,\ell_2,\ldots,\ell_s;m)$ correspond to
flags that are defined over $\bFq$.  As an example of codes
from flag varieties, consider the code $C(X;{\mathcal S})$ from
$X = {\mathcal F}(1,m-1;m)$ (that is, the
variety parametrizing flags $V_1 \subset V_2$ consisting of a line $V_1$
and a hyperplane $V_2$ containing that line).  In this case
\[
{\mathcal F}(1,m-1;m) \subset \G(1,m) \times \G(m-1,m) \simeq \proj^{m-1}\times\proj^{m-1}
\hookrightarrow \proj^{m^2-1}.
\]

\begin{theorem}\label{fvcodes}
Let ${\mathcal S}$ be the set
of all the $\bFq$-rational points on $X = {\mathcal F}(1,m-1;m)$.
Then the $C(X; {\mathcal S})$ code has parameters
\[
\left[\frac{(q^m - 1)(q^{m-1} - 1)}{(q - 1)^2}, m^2 - 1, q^{2m - 3} - q^{m-2}\right].
\]
\end{theorem}

The proof is due to Rodier and appears in \cite{r03}.  The evaluation mapping using
linear forms on $\proj^{m^2-1}$ is \emph{not} injective in this case because
the condition that $V_1 \subset V_2$ is expressed by a linear equation
in the coordinates of the Segre embedding of $\proj^{m-1}\times \proj^{m-1}$.

\subsection{Blow-ups and Del Pezzo surfaces}

Consider the surface $X = \proj^2$.  Let
\begin{equation}\label{blseq}
Y_k \to Y_{k-1} \to \cdots \to Y_1 \to Y_0 = X,
\end{equation}
be a sequence of morphisms
where for all $j$, $\pi_j : Y_j \to Y_{j-1}$ is the blow up of an $\bFq$-rational
point of the surface $Y_{j-1}$.  The result will be a surface $Y = Y_k$ containing
divisors $E_1,\ldots,E_k$ that are all contracted to a point on $X$.
Each $E_j$ is isomorphic to $\proj^1$, and
each contributes $q$ additional $\bFq$-rational points.  Therefore
\[
\# Y(\bFq) = q^2 + q + 1 + kq,
\]
which also attains the upper Weil bound for a surface with
the Betti numbers of these examples.  Whether this
construction gives interesting codes depends very much on
the the embedding of the surface $Y$
into $\proj^m$ (that is, on the linear series of divisors
forming the hyperplane sections).

One famous family of examples of such surfaces are the so-called
\emph{Del Pezzo} surfaces.
Hartshorne's text \cite{har77} and Manin \cite{m86} are good general
references for these.
By definition, a Del Pezzo surface
is a surface of degree $m$ in $\proj^m$
on which the anticanonical line bundle ${\mathcal K}^{-1}$ is
ample.  A classical result in the theory of algebraic surfaces
is that every Del Pezzo surface over an algebraically closed
field $\F$ is obtained either as the
degree 2 Veronese image of a quadric in $\proj^3$, or as follows.
Let $\ell$ be one of the integers $0,1,\ldots,6$, and take
points $p_1,\ldots,p_\ell$ in $\proj^2$ in general position (no
three collinear, and no six contained in a conic curve).  The
linear system of \emph{cubic} curves in $\proj^2$ containing
the base points $\{p_1,\ldots,p_\ell\}$
gives a rational map $\rho : \proj^2 \smallda \proj^{9 - \ell}$.
The image is a surface $X_\ell$ of degree $9 - \ell$ on which the points $p_i$
blow up to exceptional divisors $E_i \simeq \proj^1$
as in the composition of all the maps in (\ref{blseq}).  Since the
canonical sheaf on $\proj^2$ is ${\mathcal K} \simeq {\mathcal O}_{\proj^2}(-3)$,
the anticanonical
divisors are precisely the divisors in the linear system of cubics
containing $\{p_1,\ldots,p_\ell\}$.
For instance, with $\ell = 6$, $X_\ell$ is a cubic surface in $\proj^3$, and
every smooth cubic surface is obtained by blowing up some choice of points
$p_1,\ldots,p_6$.  With $\ell = 0$, the surface $X_0$ is the degree 3 Veronese
image of $\proj^2$, a surface of degree $9$ in $\proj^9$.

To get a Del Pezzo surface defined over $\bFq$, the
points $p_i$ should be $\bFq$-rational points in $\proj^2$.  This
means that the construction above can fail for certain small fields
(there may not be enough points $p_i$ in general position).  It suffices
to take $q > 4$, however in order to construct the Del Pezzo surfaces with
$0\le \ell \le 6$.

By considering the possible hyperplane sections of the Del Pezzo surface
Boguslavsky derives the following result in \cite{b98}.

\begin{theorem}\label{DPsurf}
Let $X_\ell$ be the Del Pezzo surface constructed as above and let $q > 4$. The
parameters of the $C(X_\ell)$ code are
\[
n = q^2 + q + 1 + \ell q, \quad k = 10 - \ell,
\]
and $d$ given in the following table
\[
\begin{array}{|c|c|c|c|c|c|c|c|}
\hline
\ell & 0 & 1 & 2 & 3 &\ 4\ & 5 & 6 \\ \hline
d(C(X_\ell) & q^2 - 2q & q^2 - 2q & q^2 - 2q &
q^2 - 2q + 1 &\ q^2\ & q^2 + 2q & q^2 + 4q + 1 {}^*\\ \hline
\end{array}
\]
\end{theorem}

The case $\ell = 6$ corresponds to the code from
a cubic surface in $\proj^3$.  Note the asterisk in the
table above.  In the generic case,
there are plane sections of a cubic surface consisting
of three lines forming a triangle, but
no sections consisting of three concurrent lines.
The triangle plane sections contain the maximum number of
 $\bFq$-rational points, namely $3q$.   Hence
$d(C(X_6)) = q^2 + 7q + 1 - 3q = q^2 + 4q + 1$, as claimed
in this case.  For some special configurations of points $p_i$,
however, the corresponding cubic surface will have \emph{Eckardt
points} where there is a plane section consisting of three
concurrent lines.  For those surfaces, the minimum distance
is $q^2 + 4q$ rather than $q^2 + 4q + 1$.

\subsection{Ruled surfaces and generalizations}

A \emph{ruled surface} is a surface $X$ with a mapping
$\pi : X \to C$ to a smooth curve $C$, whose fibers over all
points of $C$ are $\proj^1$'s.  Moreover, it is usually
required that $\pi$ has a section, that is, a mapping
$\sigma : C \to X$ such that $\pi \circ \sigma$ is the
identity on $C$.  For instance, over an algebraically
closed field, quadric surfaces in $\proj^3$ are isomorphic
to the product ruled surface $\proj^1 \times \proj^1$.
For background on these varieties, Chapter V of \cite{har77}
is a good reference.

Starting from a curve $C$ and a vector bundle of rank 2
(that is, a locally free sheaf of rank 2)
${\mathcal E}$ on $C$, the projective space bundle
$X = \proj({\mathcal E})$ is a ruled surface.
Conversely, every ruled surface $\pi : X \to C$
is isomorphic to $\proj({\mathcal E})$ for
some locally free sheaf of rank 2 on $C$.
Given a curve $C$ and two vector bundles on $C$,
the ruled surfaces $\proj({\mathcal E})$ and $\proj({\mathcal E}')$
are isomorphic if and only if ${\mathcal E} \simeq {\mathcal E}' \otimes {\mathcal L}$
for some line bundle ${\mathcal L}$ on $C$.  By choosing ${\mathcal L}$
appropriately, it is possible to make $H^0({\mathcal E}) \ne 0$ 
but $H^0({\mathcal E}\otimes {\mathcal M}) = 0$ whenever
${\mathcal M}$ is a line bundle on $C$ of negative degree and in this
case we say ${\mathcal E}$ is \emph{normalized}.  Then there
is a section $C_0$ of $X$ with $C_0^2 = -e$ where $e = \deg(E)$
is the degree of the divisor $E$ on $C$ corresponding to the line bundle
$\bigwedge^2 {\mathcal E}$.   If ${\mathcal E}$ is decomposable (a direct
sum of two line bundles) and normalized, then $e \ge 0$.
If ${\mathcal E}$ is indecomposable, then it is known that
$-g(C) \le e \le 2g(C) - 2$, where $g(C)$ is the genus.

Up to numerical equivalence, each divisor $D$ on $X$
is $D \sim b_1 C_0 + b_2 f$, where $f$ is a fiber of the mapping $\pi$ and
$b_1,b_2 \in \Z$.  The intersection product on divisors is
determined by the relations $C_0^2 = -e$, $C_0\cdot f = 1$,
$f^2 = 0$.  S.H.~Hansen has shown the following result.

\begin{theorem}\label{ruledsurf} Let $\pi : X \to C$ be a normalized
ruled surface with invariant $e \ge 0$.  Let $\# C(\bFq) = a$,  and
let ${\mathcal S}$ be the full set of $\bFq$-rational points on $X$.
Let ${\mathcal L} = {\mathcal O}_X(b_1 C_0 + b_2 f)$.  Then the
$C(X,{\mathcal L}; {\mathcal S})$ code has parameters
\[
[a(q + 1), \dim \Gamma(X,{\mathcal L}), d \ge n - b_2(q + 1) - (a - b_2)b_1],
\]
(provided that $b_2 < a$ and the bound on $d$ is positive).
\end{theorem}

\begin{proof}
Let $f_1,\ldots,f_a$ be the fibers of $\pi$ over the
$\bFq$-rational points of $C$.  These are disjoint curves on $X$
isomorphic to $\proj^1$, hence contain $q+1$ $\bFq$-rational
points each.  Every $\bFq$-rational point of $X$
lies on one of these lines, so $n = a(q + 1)$.  As usual, the statement for
$k$ follows if $d > 0$.  The estimate for $d$ comes from
the method of Theorem~\ref{Hansencurves} applied to the covering family
of curves $f_1,\ldots,f_a$.  In the notation of that theorem,
we have $N = q + 1$ and $\eta = (b_1 C_0 + b_2 f)\cdot f = b_1$.
At most $\ell = b_2$ of the fibers
are contained in any divisor $D$ corresponding to a global
section of ${\mathcal O}_X(b_1 C_0 + b_2 f)$ since
$D\cdot C_0 = (b_1 C_0 + b_2 f)\cdot C_0 = -e b_1 + b_2 \le b_2$.
The bound on $d$ follows immediately.
\end{proof}

The dimension of the space of global sections of ${\mathcal L}$ can be
computed via divisors on $C$ because of general facts about
sheaves on the projective space bundle $\proj({\mathcal E})$
(see \cite{har77}, Lemma V.2.4).  See the bibliographic
notes in Section~\ref{bibnotes} for more information
about these codes and for work on codes from projective
bundles of higher fiber dimension.

\section{Codes from Deligne-Lusztig Varieties}\label{DLvars}

Some of the most interesting varieties that have been used to produce
codes by the constructions of Section~\ref{constr} are
the so-called \emph{Deligne-Lusztig varieties} from representation theory.
As we will see, their description involves several of the
general processes on varieties involved in the examples above.

Let $G$ be a connected reductive
affine algebraic group over the algebraically
closure $\F$ of $\bFq$, a closed subgroup of ${\rm GL}(n,\F)$
for some $n$.  We have the $q$-Frobenius
endomorphism $F : G \to G$ whose fixed points are the $\bFq$-rational
points of $G$.

A \emph{Borel subgroup}
of $G$ is a maximal connected solvable subgroup of $G$.
A \emph{torus} is a subgroup of $G$ isomorphic to
$(\F^{\, *})^s$ for some $s$.  All Borel subgroups are conjugate,
and each maximal torus $T$ is contained in some Borel subgroup.
Let $N(T)$ be the normalizer of $T$ in $G$.
The quotient $N(T)/T$ is a finite group called the
\emph{Weyl group} of $G$.

The set ${\mathcal B}$ of all Borel subgroups of $G$
can be identified with the quotient $G/B$ for any particular $B$
via the mapping $G/B \to {\mathcal B}$ given by $g \mapsto g^{-1}Bg$.
If $w \in W$, then the Deligne-Lusztig variety associated to $w$
can be described as follows.  Let $B$ be an $F$-stable Borel
subgroup, then 
\[
X(w) = \{ x\in G : x^{-1}F(x) \in BwB\}/B \subset {\mathcal B}.
\]

\begin{theorem}\label{DLprops}
Let $w = s_1 \cdots s_n$ be a minimal factorization
of $w$ into simple reflections in $W$, the Weyl group of $G$ as above. Then
\begin{enumerate}
\item $X(w)$ is a locally closed smooth variety of pure dimension $n$.
\item The variety $X(w)$ is fixed by the action of the group $G^F$ and is defined
over $\F_{q^\delta}$, where $\delta$ is the smallest
integer such that $F^\delta$ fixes $w$.
\item The closure of $X(w)$ in ${\mathcal B}$ is the union of the
$X(s_{i_1} \cdots s_{i_r})$ such that $1 \le i_1 < i_2 < \cdots < i_r \le n$
and $X(e)$.
\end{enumerate}
\end{theorem}

We refer to \cite{ca85} for the classification of reductive
$G$ in terms of Dynkin diagrams with action of $F$.
In \cite{jh92}, J.~Hansen studied the Hermitian curves over
$\F_{q^2}$, the Suzuki curves over $\F_{2^{2n+1}}$ and the
Ree curves over $\F_{3^{2n+1}}$, all well-known maximal
curves, and all used to construct
interesting Goppa codes with very large automorphism
groups.  Hansen showed that the underlying reason these particular
curves are so rich in good properties is that they
are the Deligne-Lusztig varieties for groups $G$ for which there
is just one orbit of simple reflections in the Weyl
group under the action of $F$.  The Hermitian
curves come from groups of type ${}^2\!A_2$, the Suzuki curves
comes from the groups of type ${}^2B_2$, and the Ree curves from
the groups of type ${}^2G_2$.

It is known that there are seven cases in which there are two $F$-orbits
in the set of reflections in $W$, so taking $s_1,s_2$ from the 
distinct orbits, the Deligne-Lusztig construction with $w = s_1s_2$
leads to algebraic surfaces:
\[
A_2, C_2, G_2, {}^2\!A_3, {}^2\!A_4, {}^3\!D_4, {}^2\!F_4.
\]
One of these cases is relatively uninteresting.
In \cite{r00}, Rodier shows that the complete, smooth
Deligne-Lusztig variety $\overline{X}(s_1,s_2)$
from the group of type $A_2$ is isomorphic to the blow-up of
$\proj^2$ at all of its $\bFq$-rational points.

For the group of type ${}^2\!A_3$, however, Rodier shows
that $\overline{X}(s_1,s_2)$ is isomorphic to the blow-up
of the Hermitian surface in $\proj^3$ at its $\F_{q^2}$-rational
points.  Hence as in the discussion of the blow-ups of $\proj^2$ above,
and using Example~\ref{Hsurf}, we get a surface with $(q^3+1)(q^2+1)^2$ points.

Similarly the $\overline{X}(s_1,s_2)$ from a group of type
${}^2\!A_4$ is isomorphic to the blow-up of the complete
intersection $Y$ of the two hypersurfaces
\begin{eqnarray}\label{ci}
0 &=& x_0^{q+1} + x_1^{q+1} + \cdots + x_4^{q+1}\\
0 &=& x_0^{q^3+1} + x_1^{q^3+1} + \cdots + x_4^{q^3+1}\nonumber
\end{eqnarray}
in $\proj^4$ at the $(q^5+1)(q^2+1)$ $\F_{q^2}$-rational points on that surface.
(These are the same as the $\F_{q^2}$-rational points on the Hermitian
3-fold in $\proj^4$ defined by the first equation.)
It is easy to check that these points are all singular, and in fact
they blow up to Hermitian curves (not $\proj^1$'s) on the Deligne-Lusztig surface.
Hence the Deligne-Lusztig surface $X$ has a very large number of
$\F_{q^2}$-rational points in this case,
\[
\# X(\F_{q^2}) = (q^5+1)(q^2+1)(q^3+1).
\]
Rodier determines the structure and number of $\F_{q^\delta}$-rational
points in the $ G_2, {}^3\!D_4$, and ${}^2\!F_4$ cases as well.
Interestingly enough, his method is to realize the Deligne-Lusztig
varieties as certain subsets of flag varieties as above, where the
subspaces in the flags are related to each other using the Frobenius
endomorphism.

Rodier and S.H.~Hansen also discuss the properties
of the $C_h(X)$ codes on these varieties.  For instance in \cite{sh01},
Hansen shows the following result by relating codes on $Y$ from (\ref{ci}) and
codes on the Deligne-Lusztig surface itself.

\begin{theorem}\label{DLHansen}
Let $X$ be the Deligne-Lusztig surface of type ${}^2\!A_4$ over
the field $\F_{q^2}$.  For $1 \le h \le q^2$, there exist codes
over $\F_{q^2}$ with 
\begin{eqnarray*}
n &=& (q^5 + 1)(q^3 + 1)(q^2+1),\\
k &=& \binom{4 + h}{h} - \binom{4 + h - (q + 1)}{t - (q + 1)}, \text{and}\\
d &\ge& n - hP(q),
\end{eqnarray*}
where $P(q) = (q^3 + 1)(q^5 + 1) + (q + 1)(q^3 + 1)(q^2 - h + 1)$.
\end{theorem}

Since $P(q)$ has degree $8$ in $q$, this shows that
$d + k \ge n - O(n^{4/5})$ with $n = O(q^{10})$, some very long
codes indeed!  Hansen also considers the codes obtained from
the singular points on the complete intersection from (\ref{ci})
(that is from the Hermitian 3-fold).

\section{Connections with Other Code Constructions}\label{occ}

In this section we point out some connections between the
construction presented here and some other examples of
algebraic geometric codes related to higher dimensional
varieties in the literature.  There is a close connection
between the codes $C(X,{\mathcal L};{\mathcal S})$  and the \emph{toric codes}
constructed from polytopes or fans in $\mathbb{R}^s$ as in \cite{jh00}.  
A toric variety of dimension $s$
over an algebraically closed field $\F$ is a
variety $X$ containing a Zariski-open subset isomorphic
to the $s$-dimensional algebraic torus
$T \simeq (\F^{\,*})^s$ and on which
$T$ acts in a manner compatible with the multiplicative
group structure on $T$.  The combinatorial data in
a fan $\Sigma$ in $\mathbb{R}^s$ encodes the gluing information
needed to produce a normal toric variety $X_{\Sigma}$ from affine
open subsets of the form ${\rm Spec}(\F[S_\sigma])$
where $\F[S_\sigma]$ is a semigroup
algebra associated to the cone $\sigma$ in the fan $\Sigma$.
A polytope $P$ in $\mathbb{R}^s$ determines a normal fan $\Sigma_P$
and line bundle ${\mathcal L_P}$ on $X_{\Sigma_P}$.
The toric codes are codes $C(X,{\mathcal L}; {\mathcal S})$
for $X = X_{\Sigma_P}$, ${\mathcal L} = {\mathcal L}_P$ and
${\mathcal S} = T \cap \bFq^{\,s} = \left(\bFq^{\,*}\right)^s$.
It is not difficult to see that toric
codes are $s$-dimensional cyclic codes with certain
other properties generalizing those of Reed-Solomon
codes.

The study of decoding algorithms for
one-point algebraic geometric Goppa codes
has been unified and simplified
by the theory of \emph{order domains} discussed
in \cite{hlp98,gp02}.  The article \cite{l07}
shows how order domains can be constructed
from many of the higher dimensional varieties
discussed here.

\section{Code Comparisons}\label{comps}

It is instructive to compare
codes constructed by the methods described here
and the best currently known codes for the same $n,k$. 
We will focus on the minimum distance, although
there are many other considerations too in deciding on codes for
given applications.

All comparisons will be made by means of the online
tables of Markus Grassl, \cite{g08}.  One initial observation
is that many of the varieties $X$ that we have discussed
have \emph{so many} $\bFq$-rational
points that the $n$ achieved are far beyond
the ranges explored to date.  When no explicit codes
are known, it is still possible to make comparisons with general 
bounds.  Since the $k$ for most of the $C_h(X)$ codes
we have seen are much smaller than $n$, the \emph{Griesmer bound} 
yields some information. The usual form of the Griesmer bound 
(see \cite{hp03}) says that for an $[n,k,d]$ code over $\bFq$, 
\[
n \ge \sum_{i=0}^{k-1} \left\lceil \frac{d}{q^i}\right\rceil.
\]
Given $n,k$, this inequality can also be used to derive an upper
bound on realizable $d$ for $[n,k]$ codes that, in a sense, 
improves the Singleton bound $d \le n - k + 1$.  It should
be noted, however, that there are many pairs $n,k$ for which
there are no codes attaining the Griesmer upper bound on $d$.

We begin by noting
the following well-known fact.

\begin{theorem}
The projective Reed-Muller codes with $h = 1$
from Theorem~\ref{RMparams}
attain the Griesmer upper bound for all $m$.
\end{theorem}

This follows since $n = \# \proj^m(\bFq) = q^m + \cdots + q + 1$, $d = q^m$,
and $k = m + 1$.  

For $h > 1$, however, the presence of \emph{reducible}
forms of degree $h$, which can have many more $\bFq$-rational zeroes than 
irreducible forms (see the proof of Theorem~\ref{RMparams}), tends to
reduce the minimum distance relative to other code constructions.  This 
is true for all $q$, although the difference shows up for smaller $h$ 
the larger $q$ is.  

For instance, in the binary case, the $h = 2$ projective Reed-Muller
code with $m = 5$ has parameters $[63,21,16]$, but there are binary $[63,21,18]$ codes
known by \cite{g08}.  Similarly, with $q = 4$, the $h = 2$ projective Reed-Muller
code with $m = 3$ over $\F_4$ has parameters $[85,10,48]$, but there are
$[85,10,52]$ codes known over $\F_4$ by \cite{g08}.  In the cases that
have been explored in detail, the gap between the projective 
Reed-Muller codes and the best known codes seems to increase with 
$m$ for fixed $h$, and also for $h$ with fixed $m$
(for the cases $h < q + 1$ considered here at least).  

The minimum distance for the $C(X)$ codes from quadrics from (\ref{quadd})
also tend to be relatively close
to the Griesmer bound for their $n,k$, although the bounds grow slightly
faster than the actual $d$ as $m \to \infty$ and 
slightly better codes are known in a number of cases.  
The codes from \emph{elliptic} quadrics ($w = 0$) are superior in general
to those from hyperbolic quadrics ($w = 2$) when $m$ is odd.  This is
an interesting indication that perhaps the ``greedy'' approach of
maximizing $n = \# X(\bFq)$ does not always yield the best codes.

For example, over $\F_8$, the $C(X)$ code from a hyperbolic
quadric in $\proj^3$ has parameters $[81,4,64]$, but there
are $[81,4,68]$ codes known by \cite{g08}.  (The Griesmer
bound in this case gives $d \le 69$.)  By way of contrast,
the $C(X)$ code from an elliptic quadric has parameters
$[65,4,56]$, and this is the best possible by the Griesmer
bound.  Similar patterns hold over all of the small fields
where systematic exploration has been done.  For larger $m$,
however, it is \emph{not} always the case that $C(X)$ codes
from elliptic quadrics meet the Griesmer bound, and there are
slightly better known codes in some cases.   The $C_2(X)$
codes from quadrics seem to be similar, at least in the case
$m = 3$, where the results of Edoukou from \cite{ed08} can be applied.  
Over $\F_8$ for instance, the $C_2(X)$ code from a hyperbolic
quadric surface has parameters $[81,9,49]$, but there are $[81,9,58]$
codes known by \cite{g08}.  On the other hand, the 
$C_2(X)$ code from an elliptic quadric has parameters $[65,9,47]$, 
and this matches the best known $d$ for this $n,k$ over $\F_8$.
(The tightest known upper bound is $d \le 50$.)

The Hermitian hypersurface codes seem to be similar to those
from quadrics.  The $C(X)$ codes are quite good, coming quite
near the Griesmer bound.  For instance, the Hermitian surface
code from Theorem~\ref{hermsurfbound} over $\F_{16}$ has
parameters $[1105,4,1024]$.  This is far outside the range
of $n$ and fields for which tables are available, but
by way of comparison, $d \le 1034$ by the Griesmer bound.  
 However, the $C_2(X)$ codes are 
not as good, and the gap grows with $h$.

The codes from Del Pezzo surfaces from Theorem~\ref{DPsurf} are
interesting only for $\ell = 0$ (the case $X = \proj^2$) and 
$\ell = 6$ (the case of the cubic surface in $\proj^3$).  The
intermediate cases are quite inferior to the best known codes.  

For the other families of varieties  we have considered
(Grassmannians, flag varieties, Deligne-Lusztig varieties), once
$q$ or $m$ get even moderately large, $n$
is so huge that very little is known.  On the basis of rather limited 
evidence, the Grassmannian and flag variety codes might be 
especially good only over very small fields, though.  
For example, the $C(X)$ code from $X = \G(2,4)$ over $\F_2$ has 
parameters $[35,6,16]$, which attains the Griesmer bound.  Over $\F_3$, 
the corresponding Grassmannian code has $[130,6,81]$, but there are 
$[130,6,84]$ codes over $\F_3$ known by \cite{g08} and the Griesmer 
bound gives $d \le 84$ in this case.   

It is unrealistic to expect every code constructed
from a variety of dimension $\ge 2$ to be a world-beater.
The examples here are offered as evidence that
we still do not know how this construction can best be
applied to produce good codes.  
  
\section{Conclusion}

The study of error control codes constructed from higher
dimensional varieties is an area where it is certainly true
that we have just barely begun feeling out the lay of the
land and just barely scratched the surface of what should be 
possible.  If this survey of past work inspires further
exploration, then one of its goals will have been achieved!

\section{Bibliographic Notes}\label{bibnotes}

\noindent
{\it Section~\ref{introsec}.}  The universality of 
the Goppa construction for producing linear codes is
proved in \cite{psw91}; specifically we are referring to
Pellikaan, Shen, and van Wee's result that every linear code is \emph{weakly algebraic-geometric}:
Given $C$, there exists a smooth projective curve $X$, a set ${\mathcal S}$
of $\bFq$-rational points on $X$, and a line bundle ${\mathcal L} = {\mathcal O}(G)$
for some divisor $G$ with support disjoint from ${\mathcal S}$, such that
$C$ is isomorphic to $C(X,{\mathcal L};{\mathcal S})$ (with no restriction on the degree
of $G$).

Although very little work to date has been done on decoding methods, the 
large groups of automorphisms of some of the varieties
considered here make the \emph{permutation decoding} paradigm a 
possibility for certain of these codes.  
Some work along these lines has been done by Kroll and Vincenti,
\cite{kv05,kv08}.
  
\vskip 10pt 
\noindent
{\it Section~\ref{constr}.} Both forms of the construction
of codes from varieties (Definitions \ref{sophis} and \ref{elem})
come from \cite{tv91}, which
was the first place where this idea was described
in published form.  The form in Definition~\ref{elem} can
be made even more concrete and less algebraic-geometric
by the language of {\it projective systems} of points
and their associated codes.

\vskip 10pt
\noindent
{\it Section~\ref{params}.}
Theorem \ref{RMparams} is taken from \cite{l96}.  
It does not include the codes for $h > q$ because
the evaluation mapping is no longer injective in those cases,   
The parameters of the $C_h$ codes for $h > q$ 
have been studied by Lachaud in \cite{l90} and S\o rensen in 
\cite{so91b}.
The generalized Hamming weights $d_r$ for the Reed-Muller codes
have been studied by Heijnen and Pellikaan in \cite{hp98}.
Some ideas about finding good subcodes of the $C_2$ codes
have been presented by Brouwer in \cite{br98}.

Theorem~\ref{Hansencurves}, the following example, and
the bound using Seshadri constants in Theorem~\ref{Sesh} are all due to
S.H.~Hansen and are taken from \cite{sh01}.

The results on bounds for the minimum distance when ${\mathcal S}$
is a complete intersection come from \cite{gls05} and that
article's bibliography gives several sources for the
Cayley-Bacharach theorem and modern generalizations.  
The genesis for this was the 
observation that if ${\mathcal S}$ is a reduced complete
intersection of two cubic curves in $\proj^2$, and $\Gamma'$ is
any subset of eight of the nine points in ${\mathcal S}$, then every
cubic that contains the eight points in $\Gamma'$ also passes
through the ninth point in ${\mathcal S}$.  Related applications to coding theory
were discussed by Duursma, Renteria and Tapia-Recillas
in \cite{drt01} and J.~Hansen in \cite{jh03}.  The theorem stated here
can also be extended to yield a criterion for MDS codes.

The Weil conjectures were originally stated in \cite{weil49}
and proved in complete generality by Deligne in \cite{d74} following
three decades of work by Dwork, Serre, Artin, Grothendieck, Verdier, and many others.
Weil's paper gives a different form for middle Betti number
in (\ref{hyp}), but it can be seen that his form is equivalent to ours.
The discussion of Weil-type bounds follows Lachaud's presentation in \cite{l96}.
Because of space limitations and the significantly higher prerequisites
needed to work with the $\ell$-adic \'etale cohomology theory in any detail
in higher codimension,  we have focused only on the application of Lachaud's
results to codes from hypersurfaces.  The discussion
in \cite{l96} is considerably more general.  Edoukou has verified
S\o rensen's conjecture (see (\ref{Sorconj})) on the Hermitian surface codes
in the case $h = 2$ in \cite{ed07}.

\vskip 10pt
\noindent
{\it Section~\ref{exs}.}  The codes from quadrics have
been intensively studied since at least the 1975 article
\cite{wo75} of Wolfmann.  They are especially accessible
because so much is known about
the sets of $\bFq$-rational points on quadrics as finite
geometries; see Hirschfeld and Thas, \cite{ht91}.  The complete hierarchies
of generalized Hamming weights $d_r$ for the $C(X)$ codes
were determined independently by Nogin in \cite{n93}
and Wan in \cite{wan93}.   To aid in comparing these different
sources, we note that Wan's invariant $\delta$ is related to
Hirschfeld and Thas's (and our) character $w$ by
$\delta = 2 - w$.
The character can also be defined by
$w = 2g - m + 3$ where $g$ is the dimension of the
largest linear subspace of $\proj^m$ contained in
the quadric $X$.  Comparatively little has appeared in
the literature concerning the $C_h(X)$ codes with $h > 1$
on quadrics following the work of Aubry in \cite{a92}.
One recent article studying the $C_2(X)$ codes 
from quadrics in $\proj^3$ is Edoukou,
\cite{ed08}.

Hirschfeld and Thas also contains a wealth of information
related to the codes on Hermitian hypersurfaces.  The parameters
of the $C(X)$ codes were established by Chakravarti in \cite{chak90},
and the generalized Hamming weights were determined in
by Hirschfeld, Tsfasman, and Vladut in \cite{htv94}.

Grassmannian codes were studied first in the binary case
by C.~Ryan and K.~Ryan in \cite{r871,r872,rr90}.
The material on Grassmannian codes presented here
is taken from \cite{n96}.  In that article, Nogin
also determines the complete weight distribution for the codes
$\G(2,m)$ and shows that the generalized weights $d_r$
of the Grassmann codes meet the generalized Griesmer bound
when $r \le \max\{\ell, m - \ell\} + 1$.  More information
on the generalized weights was established by Ghorpade and
Lachaud in \cite{gl00} and
these codes are also discussed as a special case of the
code construction from flag varieties by Rodier in \cite{r03}.  This
article also gives the proof of Theorem~\ref{fvcodes}.
Codes from the Schubert varieties defined in (\ref{schub})
have been studied in \cite{chen00,gv04,gt05}.

The material on Del Pezzo surface codes is taken from
Boguslavsky, \cite{b98}.  That article also determines the complete
hierarchy of generalized Hamming weights $d_r$ for these
codes.

Codes from ruled surfaces were studied by S.H.~Hansen in \cite{sh01}
as an example of how the bound from Theorem~\ref{Hansencurves}
could be applied.  That article also addresses the
cases where the invariant $e < 0$, and presents some
examples involving ruled surfaces over the Hermitian
elliptic curve over $\F_4$.  Codes from ruled surfaces
were also considered in Lomont's thesis, \cite{lom03}.
The results for codes over ruled
surfaces have been generalized to give corresponding
results for codes on projective bundles $\proj({\mathcal E})$
for ${\mathcal E}$ of all ranks $r \ge 2$ by Nakashima in \cite{n06}.
Nakashima also considers codes on Grassmann, quadric, and 
Hermitian bundles in \cite{n05}.

Other work on codes from algebraic surfaces is contained in 
the Ph.D. theses of Lomont, \cite{lom03}, and Davis, \cite{d07}.
In addition, the unpublished preprint \cite{vz} of Voloch
and Zarzar and the article \cite{z07} adopt the interesting
approach of trying to find good surfaces for constructing
codes by limiting the presence of reducible hyperplane
sections through controlling the rank of the N\'eron-Severi group.

\vskip 10pt
\noindent
{\it Section~\ref{DLvars}}.  
Rodier's article \cite{r00} is a gold mine of information
and techniques for the Deligne-Lusztig surfaces and
Deligne-Lusztig varieties more generally.  The original
article of Deligne and Lusztig and a number
of other works devoted to this construction are referenced in the bibliography.  
The Picard group and other aspects of the finer structure
of Deligne-Lusztig varieties have been studied by
S.H.~Hansen in \cite{sh99a,sh99b,sh01}.  Hansen's thesis, \cite{sh99a}
contains chapters corresponding to the other articles here.

\vskip 10pt
\noindent
{\it Section~\ref{occ}.}  A standard reference
for the theory of toric varieties over $\C$ is Fulton's text,
\cite{f93}; the construction generalizes to fields
of characteristic $p$ with no difficulty.

\end{document}